\documentclass[11pt, a4paper]{article}

\usepackage{authblk}
\usepackage[T1]{fontenc}
\usepackage[utf8]{inputenc}

\usepackage{amsmath}
\usepackage{amssymb}
\usepackage{amsthm}
\usepackage{amsfonts}
\usepackage{graphicx} 
\usepackage[hidelinks]{hyperref}
\usepackage{enumitem}
\usepackage{mathtools}
\usepackage{cite}
\usepackage[dvipsnames]{xcolor}
\usepackage{bm}
\usepackage{csquotes}
\usepackage[ruled, vlined, linesnumbered]{algorithm2e}
\usepackage[normalem]{ulem}

\theoremstyle{plain}
\newtheorem{theorem}{Theorem}[section]
\newtheorem*{theorem*}{Theorem}
\newtheorem{proposition}[theorem]{Proposition}
\newtheorem*{proposition*}{Proposition}

\newtheorem*{corollary*}{Corollary}
\newtheorem{lemma}[theorem]{Lemma}
\newtheorem*{lemma*}{Lemma}

\theoremstyle{definition}
\newtheorem{definition}{Definition}

\newtheorem{example}{Example}
\newtheorem*{example*}{Example}

\setlist[itemize]{parsep=1pt, topsep=1pt}
\setenumerate{topsep=0pt,itemsep=-1ex,partopsep=1ex,parsep=1ex}
\setitemize{topsep=0pt,itemsep=-1ex,partopsep=1ex,parsep=1ex}
\hypersetup{colorlinks=false, linkcolor=red}
\newcommand\FF{\mathbb{F}}
\newcommand\ZZ{\mathbb{Z}}

\newcommand\RR{\mathbb{R}}
\newcommand\Zmod[1]{\mathbb{Z}/{#1}\mathbb{Z}}

\newcommand\wt{\mathrm{wt}}

\newcommand\supp{\mathrm{supp}}

\newcommand\Short{\mathrm{Short}}
\newcommand\Punct{\mathrm{Punct}}

\newcommand\calC{{\mathcal{C}}}
\newcommand\calP{{\mathcal{P}}}
\newcommand\calM{{\mathcal{M}}}

\newcommand\mydef{\coloneqq}

\title{Locality for codes over integers}

\author[1]{Giulia Cavicchioni}
\author[2]{Eleonora Guerrini}
\author[3]{Julien Lavauzelle}
\affil[1]{Fondazione Bruno Kessler, Trento, Italy\thanks{gcavicchioni@fbk.eu}}
\affil[2]{University of Montpellier, LIRMM, CNRS UMR 5506, France\thanks{guerrini@lirmm.fr}}
\affil[3]{University Paris 8, LAGA, CNRS UMR 7539, France\thanks{julien.lavauzelle@univ-paris8.fr}}

\begin{document}

\maketitle

\begin{abstract}
In this work, we study the codes over the integers with locality constraints. We introduce a weighted notion of locality over $\mathbb{Z}/q_1\mathbb{Z} \times \cdots \times \mathbb{Z}/q_n\mathbb{Z}$ and derive a Singleton-like bound for locally recoverable codes. We also propose some code constructions with locality, including integer analogs of Tamo--Barg codes.

\end{abstract}

\medskip


\medskip

\section{Introduction}
\label{sec:introduction}

 \paragraph*{Motivation.} The rise of distributed computation has highlighted the need to verify or secure the outputs of remote calculations. When outsourcing computations, servers may fail, delay responses, or, in some cases, act maliciously, producing erroneous results. Integer codes have proven effective in this setting~\cite{WatsonH66, GoldreichRS99}. The idea is to decompose a large algebraic computation, such as the determinant of a matrix with large integer entries, into many smaller but similar computations modulo small prime numbers.  After having sent each of these smaller computations and received the corresponding outputs, one is  able to retrieve the expected result applying the Chinese Remainder Theorem. 
 Moreover, if some nodes fail or return incorrect outputs, the result can still be recovered via error-correction algorithms~\cite{Mandelbaum76, Mandelbaum78a, GoldreichRS99, XiaoX15}. However, if a node fails during the computation, reconstructing the lost residue using Chinese Remainder codes typically requires contacting many servers. Indeed, due to the Chinese Remainder isomorphism, recovering a missing symbol is essentially equivalent  to reconstructing the whole data. 
 Theoretically, it is also interesting to further explore the parallels between CRT-based constructions and Reed--Solomon codes, whose connection has proven fruitful for decoding, and extend this comparison to local recovery in order to better understand their similarities and differences. In this context, it is desirable to recover lost information more efficiently. 
 This problem is closely related to the \emph{local recovery} problem~\cite{PapailiopoulosD12, GopalanHSY12}, where data is distributed and stored on servers, and should be reconstructible from a small subset in case of failures. From an algebraic perspective, a locally recoverable code enables local repair of a coordinate by accessing a small subset of other coordinates. Research in this area has focused both on bounding the minimum distance of a locally recoverable code~\cite{GopalanHSY12} and on constructing codes with optimal parameters~\cite{TamoB14,xing,jin,jin2}, that is, codes that maximize the number of codewords for a given length, distance, and locality. 

\paragraph*{Outline.} In this paper, we introduce locally recoverable codes over the integers, focusing on both deriving bounds for their minimum distance and constructing families of codes with locality. Since codes over the integers are  naturally endowed with a metric different from the Hamming metric, in Section \ref{sec:locality} we propose a new definition of locality suited for this context, and 
we state a Singleton-like bound for locally recoverable codes over the integers. Searching for good locally recoverable codes, in Section~\ref{sec:families}, we analyze two asymptotic constructions: the cartesian power of integer codes, and the (newly introduced, up to our knowledge) concatenation of integer codes. Finally, in Section~\ref{sec:tamo-barg} we propose an adaptation of the Tamo--Barg construction  over finite fields~\cite{TamoB14} to the integers. For some range of parameters, the construction approaches closely the proposed Singleton-like bound.

\section{Preliminaries}
\label{sec:preliminaries}

\subsection{Definitions and general results}

An \emph{integer code} is a subset $\calC \subseteq \Zmod{q_1}  \times \cdots \times \Zmod{q_n}\,$, where $n \ge 1$ is the \emph{length} of the code and $q_1, \dots, q_n$ are integers $\ge 2$. The size of the \emph{ambient space } $N \mydef q_1 q_2 \cdots q_n$ is also called the \emph{amplitude} of the code~\cite{GoldreichRS99}. We usually denote by $K \mydef |\calC|\le N$ its cardinality. 

In coding theory, we usually endow the ambient space with a metric which represents the loss of information failures can induce. The Hamming metric is the most usual: it represents individual errors on codeword symbols, assuming that each coordinate carries the same amount of information. Unfortunately, the latter assumption does not hold for integer codes: if some $q_i$ is much larger than the other $q_j$, the corresponding codeword symbol $c_i$ can provide much more information than the other $c_j$. Therefore, we have to consider the metric over $\Zmod{q_1} \times \cdots \times \Zmod{q_n}$  taking this into account.

Let $I \subset [1, n]$ be a subset of coordinates. The \emph{weight} of $I$ is the maximum number of symbols supported by the coordinates in $I$, that is,
\[
{\rm wt}(I) \mydef \prod_{i \in I} q_i
\]

\begin{definition}
  Let ${\bm a} = (a_1, \dots, a_n)$ and ${\bf b} = (b_1, \dots, b_n)$ be two elements of $\Zmod{q_1} \times \cdots \times \Zmod{q_n}$.
  \begin{enumerate}
  \item The \emph{support} of a vector ${\bm a}$ is the set of its nonzero coordinates, namely $\supp({\bm a}) \mydef \{ i \in [1, n] \mid a_i \ne 0 \}$.
    \item The \emph{weight} of ${\bm a}$ is ${\rm wt}({\bm a}) \mydef {\rm wt}(\supp({\bm a})) = \prod_{i: a_i \ne 0} q_i$.
      \item The \emph{weighted distance} between ${\bm a}$ and ${\bm b}$ is $D({\bm a}, {\bm b}) \mydef {\rm wt}(\supp({\bm a}-{\bm b})) = \prod_{i: a_i \ne b_i} q_i$.
  \end{enumerate}
\end{definition}


\begin{definition}
  The \emph{minimum (weighted) distance} of an integer code $\calC$ of cardinality $\ge 2$ is 
  $
  D(\calC) \mydef \min \{ D({\bm a}, {\bm b})  \,\mid\, ({\bm a}, {\bm b}) \in \calC^2,\, {\bm a} \ne {\bm b} \}.
  $
  We usually denote by $D$ the minimum distance of a code, by abuse of notation.
\end{definition}

The amplitude $N$ of a code, its cardinality $K$ and its minimum distance $D$ 
are constrained by a Singleton-like bound, as stated in Theorem~\ref{thm:singleton-integer}.

\begin{theorem}[Singleton bound for integer codes]
  \label{thm:singleton-integer}%
  Let $\calC$ be an integer code with parameters $(N, K, D)$. Then, for any set of coordinates $I \subset [1, n]$ of weight strictly less than $D$, we have ${\rm wt}(I) \le N/K$.
\end{theorem}

Note that Theorem \ref{thm:singleton-integer} can be rephrased as 
$ 
D \le \min \{ {\rm wt}(I) \mid I \subset [1, n], {\rm wt}(I) > N/K \},
$
which implies that $D < N/K$.


We now recall some classical constructions of codes from others.
\begin{enumerate}
\item The \emph{puncturing} of a code $\calC$ on a subset of coordinates $I \subset [1, n]$ is
  \[
  {\rm Punct}(\calC, I) \mydef \{ (c_i)_{i \in [1, n] \setminus I} \mid {\bm c} \in \calC \} \;\subseteq\; \prod_{i \in [1,n] \setminus I} \Zmod{q_i}\,.
  \]
  \item The \emph{shortening} of a code $\calC$ on a set of coordinates $I$, according to a tuple ${\bm a} = (a_i)_{i \in I}$, is
  \[
  {\rm Short}(\calC, I, {\bm a}) \mydef \{ {\bm c} \in \calC \mid c_i = a_i,\; \forall i \in I  \} \;\subseteq\; \prod_{i=1}^n \Zmod{q_i}\,.
  \]
  It is also usual to puncture a code that was previously shortened on $I$, since all codewords of the shortened code agree on $I$.
  Thus, we also define:
  $
  {\rm Short}^*(\calC, I, {\bm a}) \mydef {\rm Punct}({\rm Short}(\calC, I, {\bm a}), I). 
  $
\end{enumerate}




\subsection{Residue integer codes}

Let $\phi$ be the \emph{residue map}:
\[
\begin{array}{rclc} \phi:& [0, N-1]& \longrightarrow & \Zmod{q_1} \times \cdots \times \Zmod{q_n}\\
  ~ & x& \longmapsto & (x \!\!\mod q_1, \; \dots, \; x \!\!\mod q_n)
  \end{array}
\]

\begin{definition}
An integer code $\calC$ of amplitude $N$ is called a \emph{residue integer code} if there exists a subset $\calM \subset [0, N-1]$ on which $\phi$ is injective, and such that
$
\phi(\calM) = \calC\,. $

\end{definition}

In some sense, residue integer codes can be thought as integer analogues of evaluation codes over finite fields.

A very natural setting for residue integer codes arises when integers $q_i$ are pairwise coprime. The Chinese Remainer Theorem (CRT) then ensures that the residue map $\phi$ is injective over $[0, N-1]$. 
This setting leads to the construction of so-called Chinese Remainder codes (CR codes), which are counterparts of Reed--Solomon codes for integer codes \cite{GoldreichRS99}. 

\begin{definition}
  Let $K \le N = \prod_{i=1}^n q_i$, with pairwise coprime $q_i$'s. A \emph{Chinese Remainder code} is the image of $\calM \mydef [0, K-1]$ by the residue map $\phi$:
  \[
  {\rm CR}({\bm q}, K) \mydef \{ (x \!\!\mod q_1, \dots, x \!\!\mod q_n) \mid 0 \le x < K \}\,.
  \]
\end{definition}

CR codes are optimal with respect to the Singleton-like bound. 

\begin{theorem}[see \cite{GoldreichRS99}]
  The distance $D$ of the code ${\rm CR}({\bm q}, K)$ is 
  \[
  D = \min \{ {\rm wt}(I) \mid I \subset [1, n], {\rm wt}(I) > N/K \}\,.
  \]
\end{theorem}



CR codes can also be efficiently decoded, up to nearly half their minimum distance~\cite{Mandelbaum76, Mandelbaum78a,GoldreichRS99}. Decoding algorithms have also been given in other contexts, such as soft-decision decoding~\cite{GuruswamiSS00}, syndrome decoding~\cite{Li12} or interleaving~\cite{LiSN13}. The case where moduli are not coprime was also considered in~\cite{XiaoX15}.





\section{Locality for codes over the integers}
\label{sec:locality}%

In coding theory, locality refers to the ability to recover a codeword symbol by accessing only limited information from other symbols. While classical locality in the Hamming metric counts the number of queried symbols, this notion is not adequate for integer codes, where symbol sizes differ. This limitation leads to the need for a weighted notion of locality that reflects the actual recovery cost.

\begin{definition}[Coordinate locality]
  Let $\calC \subseteq \Zmod{q_1} \times \cdots \times \Zmod{q_n}$ be an integer code, and $r \in \RR^+$ be a positive real number. We say that coordinate $i \in [1, n]$ has locality $r$ in $\calC$ if there exists $S \subset [1, n] \setminus \{ i \}$ such that
  \begin{enumerate}
  \item \emph{(locality)} $\wt(S) \le \wt(\{i\})^r = q_i^r$;
  \item \emph{(recoverability)} $|\Punct(\mathcal{C}, [1,n] \setminus S)| =  |\Punct(\mathcal{C}, [1,n] \setminus (S \cup \{ i\}))|$.
  \end{enumerate}
  The subset $S$ is called a \emph{helper set} for the coordinate $i$.
\end{definition}
The recoverability condition requires that the symbol $c_i$ of any codeword
$c \in \mathcal{C}$ is uniquely determined from the symbols $c_j$, $j \in S$.
The locality condition then bounds the ratio between $\log_2 \wt(S)$, the number of
bits accessed during the recovery, and $\log_2 \wt(\{i\})$, the number of bits to be
recovered.

As in the Hamming metric, we do not consider degenerate  codes where a coordinate carries no parity information and cannot be recovered from the others. We formalize this as follows.




\begin{definition}[Information set]
    Let $I \subset [1, n]$ be a subset of the coordinates. If  $|\Punct(\mathcal{C}, I)| = | \mathcal{C}|$ then $I$ is an \emph{information set} for $\mathcal C$. An information set is said \emph{minimal} if all its proper subsets are not information sets for $\calC$.
\end{definition}

\begin{definition}[Non-degenerate integer code]
  An integer code $\calC$ of length $n$ is \emph{non-degenerate} if, for all $i \in [1, n]$, the subset $[1, n] \setminus \{ i \}$ is an information set for $\calC$.
\end{definition}

In the rest of this paper, we only consider non-degenerate codes.

\begin{definition}[Code locality]
  An integer code $\calC \subseteq \Zmod{q_1} \times \cdots \times \Zmod{q_n}$ has locality  $r  \in \RR^+$ if each coordinate $i \in [1, n]$ has locality $r$ in $\calC$. We also say that $\calC$ is $r$-locally recoverable (LR).
\end{definition}

\begin{example}
  \label{ex:repetition}%
  Assume that $q_1 < \dots < q_n$, and consider the \enquote{repetition} integer code 
  \[
  \calC := {\rm CR}({\bf q}, q_1) = \{ (x, x, \dots, x) \mid 0 \le x < q_1 \} \subseteq \Zmod{q_1} \times \cdots \times \Zmod{q_n}\,.
  \]
  Each symbol of a codeword $(x, x, \dots, x) \in \calC$ can be recovered by reading any other symbol of the codeword. Therefore, if we want to use a helper set as small as possible, we can query the coordinate with lowest weight:
  \begin{itemize}
    \item for $i = 1$, then we can choose $S = \{ 2 \}$ as an helper set. Thus, coordinate $1$ has locality $\log(q_2)/\log(q_1) > 1$; 
    \item for $i \ne 1$, then we can choose $S = \{ 1 \}$ as an helper set. Thus, coordinate $i$ has locality $\log(q_1)/\log(q_i) < 1$.
  \end{itemize}
  Hence, the repetition code $\calC$ has locality $r = \log(q_2)/\log(q_1)$.
\end{example}

\begin{example}
  Let $q_1,\dots, q_n$ be pairwise coprime and consider the Chinese Remainder code of cardinality $K$.  Each coordinate in $\mathrm{CR}(\mathbf q, K)$ has a locality of at least \( \frac{\log K}{\log q_{\max}} \) where $q_{\max} \coloneqq \max_i q_i$. To understand this, let us fix a coordinate $i \in [1, n]$ and consider a set of coordinates $S \subseteq [1,n] \setminus \{i\}$  of weight less than $K$. Since $\mathcal{C}$ contains $K$ distinct codewords, there must exist at least two codewords that agree on $S$. Hence, $S$ cannot serve as a helper set for  $i$. On the other hand, Chinese Remainder Theorem implies that any set of coordinates $T \subseteq [1,n]$ of weight at least $K$ forms an information set: the values of the coordinates in $T$ uniquely determine the codeword. Therefore, a missing coordinate $i$ can only be recovered by accessing a set of total weight at least $K$, effectively reconstructing the entire codeword. This shows that Chinese Remainder codes exhibit maximal locality: recovering a single symbol requires all the information of the codeword. This behavior is analogous to MDS codes in the Hamming metric, where no coordinate can be retrieved from any subset of coordinates that does not already form an information set. 



\end{example}

\subsection{A Singleton-like bound for LR codes over integers}


In this section, we derive a bound on the minimum distance of an LR code over the integers as a function of its locality and parameters, adapting techniques originally developed for Hamming-metric codes~\cite{PapailiopoulosD12, GopalanHSY12}. A key step is the identification of conditions under which locality is preserved after a shortening. 


\begin{lemma}
  \label{lem:shortening-locality}
  Let $\calC$ be an integer code with locality $r$. Let $i \in [1, n]$ and fix $S$ a helper set for $i$ of weight $\le q_i^r$. Then, 

  \begin{enumerate}
  \item For all ${\bm a} = (a_j)_{j \in S \cup \{i\}}$, the shortened code $\Short(\calC, S \cup \{ i\}, {\bm a})$ has locality at most $r$;
  \item There exists an ${\bm a} = (a_j)_{j \in S}$ such that $|\Short(\calC, S, {\bm a})| \ge \frac{|\calC|}{\wt(S)}$.
    \end{enumerate}
\end{lemma}

\begin{proof}
  The first point follows from the fact that if $\mathcal C' \subseteq \mathcal C$, then any helper set for a coordinate $i$ in $\mathcal C$ remains a helper set for $i$ in $\mathcal C'$. The second claim follows from the pigeonhole principle.
\end{proof}

\begin{theorem}[LRC bound for integer codes]
  \label{thm:singletion-intLRC}%
  Let $\calC$ be an $(N, K, D)$ integer code over $\Zmod{q_1} \times \dots \times \Zmod{q_n}$ with locality $r$. Then,
  \begin{equation}\label{eq:lrc} D < N \left(\frac{\max \{ q_i \}}{K}\right)^{1+1/r}\,.
  \end{equation}
 
\end{theorem}

\begin{proof}
We use Algorithm~\ref{algo:iterative-shortening} to iteratively construct a sequence of shortened subcodes of $\mathcal{C}$, whose parameters can be bounded accordingly.
\begin{algorithm}
  \KwData{An integer code $\calC$ over $\Zmod{q_1} \times \dots \times \Zmod{q_n}$}
  \KwResult{A shorter integer code}

  $I \leftarrow \varnothing$, \ 
  $q \leftarrow \min \{ q_j : j \notin I \}$\\
  \While{$|\calC| \ge q$}{
    Pick $i \in [1,n]\setminus I$. Find a helper set $S_i$ for $i$ and a tuple ${\bm a} = (a_j)_{j \in S_i \cup \{i\}}$, such that $|{\rm Short}^*(\calC, S_i \cup \{i\}, {\bm a})| \ge \frac{|\calC|}{\wt(S_i)}$.\\
    \If{$|{\rm Short}^*(\calC, S_i \cup \{i\}, {\bm a})| \ge 2$}{
     $I \leftarrow I \cup S_i \cup \{ i \}$.\\
       $q \leftarrow \min \{ q_j : j \notin I \}$\\
    $\calC \leftarrow {\rm Short}^*(\calC, S_i \cup \{i\}, {\bm a})$.\\
    }
    \Else{
      Find a proper subset $T_i$ of $S_i$, and a  tuple ${\bm b} = (b_j)_{j \in T_i \cup \{i\}}$, such that $2 \le |{\rm Short}^*(\calC, T_i \cup \{i\}, {\bm b})| < q$.\\
       $\calC \leftarrow {\rm Short}^*(\calC, T_i \cup \{i\}, {\bm b})$.\\
    }
    
  }
  {\bf Return} $\calC$
  \caption{\label{algo:iterative-shortening}Iterative construction of shortened subcodes}
\end{algorithm}

Let $t \ge 1$ denote the number of iterations of the {\tt while} loop of  Algorithm~\ref{algo:iterative-shortening}. For $1 \le j \le t$, let $\mathcal{C}_j$ be the code stored in the variable $\mathcal{C}$ \emph{upon entering} the loop for the $j$-th iteration. Also denote by $\mathcal{C}_{t+1}$ the code returned by the algorithm. Notice that the {\tt else} branch produces a code $\calC$ of size $<q \coloneqq\min_i q_i $, hence we enter this branch only for the last iteration of the {\tt while} loop.  Let also $(N_j, K_j, D_j)$ be the parameters of the code $\mathcal{C}_j$; by defintion $\mathcal{C}_1 = \mathcal{C}$ and $(N_1, K_1, D_1) = (N, K, D)$.

First, for $j = 1, \dots, t-1$, we can easily check that the code $\calC_j$ has amplitude $N_j = N_{j+1} \cdot {\rm wt}(S_{i_j} \cup \{ i_j \})$, cardinality $K_j \le K_{j+1} \cdot {\rm wt}(S_{i_j})$, distance $D_j \le D_{j+1}$ and locality at most $r$.

For $j = t$, we have two cases. If the algorithm enters the {\tt if} statement in the last iteration, then the previous claims on $\calC_j$ remain valid for $\calC_{t+1}$.
    Moreover $K_{t+1} < q_{\max} := \max_\ell q_\ell$.
Otherwise ({\tt else} branch), the code $\calC_t$ has amplitude $N_t = N_{t+1} \cdot {\rm wt}(T_{i_t} \cup \{ i_t \})$, cardinality $K_t \le K_{t+1} \cdot {\rm wt}(T_{i_t})$, and distance $D_t \le D_{t+1}$.

The claim on $K_j$ follows from Lemma~\ref{lem:shortening-locality}. Note that for $j = t$, in the \enquote{{\bf else}} branch, $T_{i_t}$ is no longer a helper set, so the weight of $\{i_t\}$ cannot be removed in the lower bound for $K_{t+1}$. The claim on $N_j$ follows from a puncturing argument, while the statements on minimum distance and locality hold because $\mathcal{C}_j \subseteq \mathcal{C}_{j+1}$. Combining these results, we get
  \[
  q_{\max} > K_{t+1} \ge \frac{K_t}{{\rm wt}(T_{i_t} \cup \{i_t\})} \ge \frac{K}{{\rm wt}(S_{i_1}) \cdots {\rm wt}(S_{i_{t-1}}) \cdot {\rm wt}(T_{i_t} \cup \{i_t\})}
  \]
  and
  \[
  N_{t+1} = \frac{N_t}{{\rm wt}(T_{i_t} \cup \{ i_t \})} = \frac{N}{{\rm wt}(S_{i_1} \cup \{i_1\}) \cdots {\rm wt}(S_{i_{t-1}} \cup \{i_{t-1}\}) \cdot {\rm wt}(T_{i_t} \cup \{ i_t \}) }. 
  \]
  More easily $D_{t+1} \ge \dots \ge D_1$, and the Singleton bound applied on $\calC_{t+1}$ gives
  \[
  D_1 \le D_{t+1} < \frac{N_{t+1}}{K_{t+1}} \le \frac{N}{K} \cdot \frac{1}{{\rm wt}(\{i_1\}) \cdots {\rm wt}(\{i_{t-1}\})}\,.
  \]

  Then, the recoverability property ensures that codewords in $\calC$ are fully determined by $S_{i_1} \cup \dots \cup S_{i_t}$  with a last residue. Hence, $K_1 \le {\rm wt}(S_{i_1}) \cdots {\rm wt}(S_{i_{t-1}}) \cdot {\rm wt}(S_{i_t}) \cdot q$,
  which implies by the locality property, that
  \[
  K_1 \le \Big({\rm wt}(\{i_1\}) \cdots {\rm wt}(\{i_{t-1}\}) \cdot {\rm wt}(\{i_t\})\Big)^r q \le \Big({\rm wt}(\{i_1\}) \cdots {\rm wt}(\{i_{t-1}\})\Big)^r   q^{1+r}.
  \]
  Thus, the desired result follows from 
  \[
  \frac{1}{{\rm wt}(\{i_1\}) \cdots {\rm wt}(\{i_{t-1}\})} \le \frac{q^{1+1/r}}{K_1^{1/r}} \ .
  \]
  
\end{proof}

\section{Some families of codes with locality}
\label{sec:families}%
In this section, we present two constructions of integer codes with locality and analyze their parameters.

\paragraph*{Product codes.} Let $\calC$ and $\calC'$ be two integer codes with respective parameters $(N, K, D)$ and $(N', K', D')$. Their \emph{Cartesian product}
\[
\calC \times \calC' \mydef \{ (c_1, \dots, c_n, c'_1, \dots, c'_m) \mid {\bm c} \in \calC, {\bm c'} \in \calC' \}
\]
has amplitude $N N'$, cardinality $K K'$, and minimum distance $\min\{ D, D' \}$. 

\begin{lemma}
  If $\calC$ has locality $r$ and $\calC'$ has locality $r'$, then $\calC \times \calC'$ has locality $\le \max\{ r, r' \}$.
\end{lemma}


For $m \ge 1$, consider the $m$-th power $\calC^m$ of a Chinese Remainder code $\calC$ with locality $r$. The code $\calC^m$ has parameters $(N^m, K^m, D)$ and locality $\le r$. As $m \to \infty$, this yields an asymptotic family of codes with  locality $r$, information rate $\frac{\log(K)}{\log(N)}$, and vanishing relative distance $\frac{\log(D)}{m\log(N)}$.

\paragraph*{Concatenated codes.}
In the Hamming metric, random concatenated codes~\cite{Forney66} yield families of locally recoverable codes with good parameters~\cite{CadambeM15}. We now  extend this notion  to integer codes.


\begin{definition}
Let $\calC_{\rm out} \subseteq \Zmod{q_1} \times \cdots \times \Zmod{q_n}$ be an integer code, and let $q_{\max} = \max\{q_i\}$. Let $\calC_{\rm in} \subseteq \Zmod{\ell_1} \times \cdots \times \Zmod{\ell_m}$ be an integer code of cardinality at least $q_{\max}$, and fix an injective map $\phi: [0, q_{\max}-1] \to \calC_{\rm in}$.  
The \emph{concatenation of $\calC_{\rm out}$ and $\calC_{\rm in}$} is 
\[
\calC_{\rm in} \circ_\phi \calC_{\rm out} \coloneqq \big\{ (\phi(c_1), \dots, \phi(c_n)) \mid  {\bm c} \in \calC_{\rm out} \big\} \subseteq (\Zmod{\ell_1} \times \cdots \times \Zmod{\ell_m})^n.
\]
\end{definition}



\begin{proposition}
Let $\calC_{\rm out} \subseteq \Zmod{q_1} \times \cdots \times \Zmod{q_n}$ and  $\calC_{\rm in} \subseteq \Zmod{\ell_1} \times \cdots \times \Zmod{\ell_m} $
be integer codes with parameters $(N_{\rm out}, K_{\rm out}, D_{\rm out})$ and $(N_{\rm in}, K_{\rm in}, D_{\rm in})$, respectively.
Let $q_{\rm min} = \min\{q_1, \dots, q_n\}$ and $q_{\rm max} = \max\{q_1, \dots, q_n\}$. Given an injective map $\phi : [0, q_{\rm max}-1] \to \calC_{\rm in}$, denote by $\calC_{\rm in} \circ_\phi \calC_{\rm out}$ the concatenation of $\calC_{\rm out}$ with $\calC_{\rm in}$. The code $\calC_{\rm in} \circ_\phi \calC_{\rm out}$ has:
\begin{enumerate}[label=(\roman*)]
  \item amplitude $N = N_{\rm in}^n$;
  \item cardinality $K = K_{\rm out}$;
  \item minimum distance satisfying $\log D \;\ge\; (\log D_{\rm in})(\log_{q_{\rm min}} D_{\rm out})$;
  \item locality $r_{\rm in}$ where $r_{\rm in}$ is the locality of the inner code $\calC_{\rm in}$.
\end{enumerate}
\end{proposition}

\section{A Tamo--Barg construction for LRC over integers}
\label{sec:tamo-barg}%

In this section, we extend the Tamo--Barg construction~\cite{TamoB14} to codes over the integers. Recall that the Tamo--Barg construction, originally introduced for finite fields and  generalized to rings~\cite{cavicchioni2025class}, defines codes of the form
\[
{\rm TB}_{r,s}(\calP, g) \mydef \langle {\rm ev}_\calP(x^ig(x)^j) \mid 0 \le i < r, 0 \le j < s \rangle \subseteq \FF_q^{(r+1)s},
\]
where $\calP = \sqcup_i \calP_i \subseteq \FF_q$ and $g(x) \in \FF_q[x], \deg(g) = r+1$, are chosen such that each $|g(\calP_i)| = 1$. We propose an analogous construction of Tamo--Barg codes for integer codes. The analogue of the evaluation map $f \mapsto f(x_i)$ is the residue map $A \mapsto A \mod q_i$; the analogue of the Tamo--Barg constraint over $g(x)$ is that some integer $G$ must be constant when reduced modulo products of disjoint subsets of the $q_i$'s. This leads us to the following definition:
\[
  \begin{aligned}
    {\rm iTB}_{\bm A}({\bm q}, G, D) \mydef \{ c \mydef \textstyle\sum_j a_j G^j \!\!\!\mod N \mid \; & 0 \le c < N/D \text{ and } \\[-0.25em]
    &\quad 0 \le a_j < A_j,  \forall 1 \le j \le s
  \},
  \end{aligned}
  \]
where ${\bf q} \mydef (q_1, \dots, q_n)$, $Q_i \mydef q_{(r+1)(i-1)+1} \cdots q_{(r+1)i}$,  ${\bm A} \mydef (A_1, \dots, A_s)$ and $G \ge 1$ is chosen such that $G \!\!\mod Q_i$ is a \enquote{small} integer $\lambda_i$. Indeed, as we will see in Theorem \ref{thm:TB_param}, the minimum distance of an iTB code scales inversely with $\max_i \lambda_i$. As in the classical TB construction, the codes we introduce arise as subcodes of the Cartesian product of CR codes, each of size at most \(\frac{\max_i Q_i}{\max_i q_i}\).

\begin{example}[Illustrative example integer-TB codes]
  \label{ex:tamo-barg}%
Consider $s=2$ groups of size $r+1 = 3$.
Let \(q_1,\ldots,q_6\) be pairwise coprime integers and set \(q \coloneqq \max_i q_i\).
Define $Q_1 \coloneqq q_1 q_2 q_3$, and $ Q_2 \coloneqq q_4 q_5 q_6$. Moreover,
let \(N \coloneqq Q_1 Q_2\) and \(Q \coloneqq \max\{Q_1,Q_2\}\).

By Bézout’s identity, there exist integers \(u_1,u_2\) such that
$1 = u_1 Q_1 + u_2 Q_2$.
For integers \(\lambda_1,\lambda_2 \ge 1\), define
\begin{equation}\label{eq:G}
    G \coloneqq u_1 \lambda_1 Q_1 + u_2 \lambda_2 Q_2\ . 
\end{equation}
Choosing \(\lambda_1,\lambda_2\) such that \(G \ge 0\) ensures $0 \le G \le u \lambda Q$, 
where \(\lambda \coloneqq \max\{\lambda_1,\lambda_2\}\)  and \(u \coloneqq \max\{|u_1|,|u_2|\}\).

Codewords in the integer-TB code are of the form
\(c\coloneqq a + bG\), with \(a,b \ge 0\) satisfying suitable constraints.
\begin{itemize}
\item \emph{Local recovery:} we impose that $c \!\!\mod Q_1$ lies in a non-trivial CR code of amplitude $Q_1$, hence we require $0 \le c \!\!\mod Q_1 = a + \lambda_1 b < \frac{Q_1}{q}$; a similar constraint is required for $c \!\!\mod Q_2$.
\item \emph{Global minimum distance:} we prescribe a distance for the code, as a subcode of ${\rm CR}({\bm q}, \frac{N}{D})$ for some target distance $D \ge 1$, which is typically set as a product of the moduli \(q_i\).
\end{itemize}
In this setting, an \emph{integer-TB} code $\calC\subseteq\ZZ/q_1\ZZ \times \cdots \times \ZZ/q_6\ZZ$ is defined as:
\[
\mathcal C \coloneqq 
\left\{
a + bG \;\middle|\;
0 \le a + bG < \frac{N}{D}, \;
0 \le a + b\lambda < \frac{N}{Qq}
\right\}.
\]
\end{example}

Since the minimum distance is fixed and the location is determined by the choice of $q_i$, the main challenge is to select $G$ to maximize the code size.
Experiments show that, for a given locality, integer-TB codes can achieve either larger distance (but smaller cardinality) or larger cardinality (but smaller distance) compared to the Cartesian product of CR codes.

In the setting of Example~\ref{ex:tamo-barg} with $q_1 \simeq \dots \simeq q_6 \simeq q$, we were able to produce codes of amplitude $N \sim q^6$, cardinality $K = c q^2$ where $c > 2$ is a small constant, distance $D \sim q^3$ and locality $r \simeq 2$, see table below.
{\scriptsize
\[
\begin{array}{c|c|c|c|c}\label{tab:paramTB}
  (q_1, \dots, q_6) & (\lambda_1, \lambda_2) & D & K = |\calC| & \text{bound on } D~\eqref{eq:lrc} \\ \hline\hline
  (37, 41, 47, 23, 31, 50) & (3, 1) & 35\,650 = q_4q_5q_6 & \ge 15\,640 \sim 6 q_6^2& < 556\,474 \\
  (59, 71, 79, 61, 67, 73) & (252, 1) & 298\,351 = q_4q_5q_6 & \ge 21\,252 \sim 4 q_6^2 &< 23\,250\,800 \\
  (97, 107, 113, 101, 103, 109) & (161, 1) & 1\,133\,927 = q_4q_5q_6 & \ge 30\,654 \sim 3 q_6^2 & < 242\,963\,524\\
\end{array}
\]
}

In general, deriving an explicit formula for the size of integer-TB codes is challenging. The following is a lower bound on the size of integer-TB codes.

\begin{theorem}\label{thm:TB_param}
Let $q_1,\ldots,q_6$ be pairwise coprime integers with $q_{\max} \coloneqq \max_i q_i$, and let $Q_1, Q_2, N, Q$ be defined as above. For given integers $\lambda_1, \lambda_2 \ge 1$, let $G$ be defined according to~\eqref{eq:G}. Then the $\rm iTB$ code $\mathcal{C}$ has locality $
r \le \frac{\log N/Q}{\log q_{\max}}-1$.
Moreover, if $N \approx Q^2$, the code size $K$ satisfies  
\[
K \ge \frac{1}{u \lambda }\left( \frac{N}{D}-q_{\max}^{r-1}\right).
\]
\end{theorem}
\section{Conclusions}
In this work, we introduced a notion of locality for codes over the integers endowed with a weighted metric. We derived a bound on the minimum distance of integer locally recoverable codes as a function of its amplitude, size, and locality. We also proposed some families of integer codes with locality, including Cartesian powers and concatenated constructions, and presented an adaptation of the Tamo--Barg construction to the integer setting. While deriving explicit formulas for the minimum distance of these codes remains challenging, we showed that this family contains almost-optimal codes. Future work includes a finer analysis of the size and minimum distance of integer Tamo--Barg codes and determining the tightness of the bound for locally recoverable codes.
\section*{Acknowledgement}
The first  author conducted the majority of this research at the Institute of Communication and Navigation, German Aerospace Center, Oberpfaffenhofen,
Germany. The first author is currently affiliated with the Fondazione Bruno Kessler, Trento, Italy, and is supported by Ministero delle Imprese e del Made in Italy (IPCEI Cloud DM 27 giugno 2022 - IPCEICL-0000007).




\bibliographystyle{alpha}
\bibliography{biblio.bib}

\end{document}